\documentclass{llncs}

\usepackage[utf8]{inputenc}
\usepackage{amsmath,amssymb}
\usepackage{xspace}
\usepackage{dsfont}
\usepackage[capitalise]{cleveref}
\usepackage[ruled,vlined]{algorithm2e}
\usepackage{graphicx}
\usepackage[caption=false,font=footnotesize]{subfig} 
\usepackage[T1]{fontenc}
\usepackage{longtable}
\usepackage{url}
\usepackage{color}
\usepackage{soul}

\usepackage{xcolor}
\usepackage[inline,nomargin,author=]{fixme}
\fxsetface{inline}{\rmfamily}

\newcommand{\ori}[1]{\smash{\overrightarrow{\!#1}}}

\newcommand{\loghole}{log-hole\xspace}

\newcommand{\hide}[1]{}

\title{Enumerating Cyclic Orientations of a Graph}

\author{Alessio Conte\inst{1} \and Roberto Grossi\inst{1} \and Andrea Marino\inst{1} \and Romeo Rizzi\inst{2}}
\institute{Universit\`a di Pisa, \email{conte,grossi,marino@di.unipi.it} \and Universit\`a di Verona, \email{rizzi@di.univr.it}}

\begin{document}
\maketitle

\begin{abstract}
  Acyclic and cyclic orientations of an undirected graph have been
  widely studied for their importance: an orientation is acyclic if it
  assigns a direction to each edge so as to obtain a directed acyclic
  graph (DAG) with the same vertex set; it is cyclic otherwise. As far
  as we know, only the enumeration of acyclic orientations has been
  addressed in the literature. In this paper, we pose the problem of
  efficiently enumerating all the \emph{cyclic} orientations of an
  undirected connected graph with $n$ vertices and $m$ edges,
  observing that it cannot be solved using algorithmic techniques
  previously employed for enumerating acyclic orientations.  We show
  that the problem is of independent interest from both combinatorial
  and algorithmic points of view, and that each cyclic orientation can
  be listed with $\tilde{O}(m)$ delay time. Space usage is $O(m)$ with
  an additional setup cost of $O(n^2)$ time before the enumeration
  begins, or $O(mn)$ with a setup cost of $\tilde{O}(m)$ time.
\end{abstract}

\section{Introduction}
\label{sec:intro}
Given an undirected graph $G(V,E)$ with $n=|V|$ vertices and $m=|E|$
edges, an orientation transforms $G$ into a directed graph $\ori{G}$
by assigning a direction to each edge. That is, an orientation of $G$
is the directed graph $\ori{G}(V,\ori{E})$ such that the vertex set
$V$ is the same as $G$, and the edge set $\ori{E}$ is an orientation
of $E$: exactly one direction between $(u,v)\in \ori{E}$ and
$(v,u)\in \ori{E}$ holds for any undirected edge $\{u,v\}\in E$.  An
orientation $\ori{G}$ is \emph{acyclic} when $\ori{G}$ does not
contain any directed cycles, so $\ori{G}$ is a directed acyclic graph
(DAG); otherwise we say that the orientation $\ori{G}$ is cyclic.

Acyclic orientations of undirected graphs have been studied in
depth. Many results concern the number of such orientations:
Stanley~\cite{stanley87} shows how, given a graph $G$ and its
chromatic number $\chi$, the number of acyclic orientations of $G$ can
be computed by using the chromatic polynomial (a special case of
Tutte's polynomial). Other results rely on the so called acyclic
orientation game: Alon et al.~\cite{alon1995acyclic} inquire about the
number of edges that have to be examined in order to identify an
acyclic orientation of a random graph $G$; Pikhurko~\cite{CPC:6776456}
gives an upper bound on this number of edges in general graphs.  The
counting problem is known to be \#P-complete~\cite{linial1986hard} and
enumeration algorithms that list all the acyclic orientations of a
graph are given in \cite{Barbosa199971} and \cite{Squire1998275}.

We consider \emph{cyclic} orientations, which have been also studied
from many points of view. Counting them is
co-\#P-complete~\cite{linial1986hard}. In Fisher et
al.~\cite{Fisher199773}, given a graph G and an acyclic orientation of
it, the number of \emph{dependent edges}, i.e. edges generating a
cycle if reversed, has been studied.
This number of edges implicitly gives a hint on the number of
cyclic orientations in a graph. 
In~\cite{Richardson:1996:DAD:2074284.2074338} an algorithm has been
given to make inference about causal structure in cyclic
graphs. In~\cite{Spirtes:1995:DCG:2074158.2074214} directed cyclic
graphs are used to model economic processes, and an algorithm is given
to characterize conditional independence constraints of these
processes.

In this paper we address the problem of enumerating all the cyclic
orientations of a graph.
\begin{problem}
  Enumerating the set of all the directed graphs $\ori{G}$ that are
  cyclic orientations of an undirected graph $G$.
\label{problem1}
\end{problem}

We analyze the cost of an enumeration algorithm for
Problem~\ref{problem1} in terms of its \emph{setup} cost, meant as the
initialization time before the algorithm is able to lists the
solutions, and its \emph{delay} cost, which is the worst-case time
between any two consecutively enumerated solutions
(e.g. \cite{JohnsonP88}). We are interested in algorithms with
guaranteed $\tilde{O}(m)$ delay, where the $\tilde{O}$ notation ignores
polylogs.

A naive solution to Problem~\ref{problem1} uses the fact that
enumeration algorithms exist for listing acyclic
orientations~\cite{Barbosa199971,Squire1998275}. It enumerates the
cyclic orientations by difference, namely, by enumerating all the
$2^m$ possible edge orientations and removing the $\alpha$ acyclic
ones. However, this solution does not guarantee any polynomial delay, as
the number $\beta = 2^m-\alpha$ if cyclic orientations can be much
larger or much smaller than the number $\alpha$ of acyclic ones.  For
example, a tree with $m$ edges has $\alpha=2^m$ and $\beta=0$. On the
other extreme of the situation, we have cliques. An oriented clique is
also called a \textit{tournament}, and a \textit{transitive
  tournament} is a tournament with no cycles. A clique of $n$ nodes
(and $m=\frac{n\cdot(n-1)}{2}$ edges) can generate $2^m$ different
tournaments, out of which exactly $\alpha=n!$ will be transitive
tournaments~\cite{moon1968topics}. As $2^m$ grows faster than $n!$, we
have that the ratio $\alpha/\beta$ tends to 0 for increasing $n$,
where $\beta = 2^m-\alpha$.

To the best of our knowledge, an enumeration algorithm for
Problem~\ref{problem1} with guaranteed $\tilde{O}(m)$ delay is still
missing. We provide such an algorithm in this paper, namely, listing
each cyclic orientation with $\tilde{O}(m)$ delay time. Space usage is
$O(m)$ memory cells with a setup cost of $O(n^2)$ time, or $O(mn)$
memory cells with a setup cost of $\tilde{O}(m)$ time. Interestingly, 
Problem~\ref{problem1} reveals to be a rich source for enumerations techniques, and our solution
offers new combinatorial and algorithmic techniques when compared to
previous work on the enumeration of acyclic
orientations~\cite{Barbosa199971,Squire1998275}.

In the following, for the sake of clarity, we will call \emph{edges}
the elements of $E$ (undirected graph) and \emph{arcs} the elements of
$\ori{E}$ (directed graph).  We will assume that the graph in input
$G$ is connected and we will denote as $n=|V|$ and $m=|E|$
respectively its number of nodes and edges.

The paper is organized as follows. Section~\ref{sec:overview}
 gives an overview of our enumeration algorithm. Section~\ref{sec:setup} describes the initialization steps, and shows how to reduce the problem from the input graph to a suitable multi-graph that guarantees to have a chordless cycle (hole) of logarithmic size. The latter is crucial to obtain the claimed delay.
Section~\ref{sec:enum} shows how to enumerate in the multigraph and obtain the cyclic orientations for the input graph. Section~\ref{sec:absorb} describes how to absorb the setup cost using more space. Finally, some conclusions are drawn in Section~\ref{sec:conclusions}.

\section{Algorithm overview}
\label{sec:overview}

The intuition behind our algorithm for an undirected connected graph $G(V_G,E_C)$ is the following one. Suppose that $G$ is cyclic, otherwise there are no cyclic orientations. Consider one cycle\footnote{This will actually be a chordless cycle of logarithmic size (called \loghole).} $C(V_C,E_C)$ in $G$: we can orient its edges in two ways so that the resulting $\ori{C}$ is a directed cycle. At this point, \emph{any} orientation of the remaining edges, e.g. those in $E_G \setminus E_C$, will give a cyclic orientation of $G$. Thus, the interesting cases are when the resulting $\ori{C}$ is not a directed cycle. 

The idea is first to generate all possible orientations of the edges in $E_G \setminus E_C$, and then assign some suitable orientations to the edges in $E_C$. This guarantees that we have at least two solutions for each orientation of $E_G \setminus E_C$, namely, setting the orientation of $E_C$ so that $\ori{C}$ is one of the two possible directed cycles. Yet this is not enough as we could have a cyclic orientation even if $\ori{C}$ is \emph{not} a directed cycle. 

Therefore we must consider some cases. One easy case is that the orientation of $E_G \setminus E_C$ already produces a directed cycle: any orientation of $E_C$ will give a cyclic orientation of $G$. Another easy case, as seen above, is for the two orientations of $E_C$ such that $\ori{C}$ is a directed cycle:  any orientation of $E_G \setminus E_C$ will give a cyclic orientation of $G$. It remains the case when the orientations of both $E_G \setminus E_C$ and $E_C$ are individually acyclic: when put together, we might have or not a directed cycle in the resulting orientation of $G$. To deal with the latter case, we need to ``massage'' $G$ and transform it into a multigraph as follows. We refer the reader to Algorithm~\ref{alg:structure}. 

\textbf{Algorithm setup} is performed as described in Section~\ref{sec:setup}. We preprocess $G$ with dead-ends removal and edge chain compression. The result is an undirected connected multigraph $M(V_M,E_M)$, where the edges are labeled as \emph{simple} and \emph{chain}. After that we find a chordless cycle of logarithmic size $C$ in $M$, and remove $E_C$ from $E_M$, obtaining the labeled multigraph $M'(V_M,E'_M)$, where $E'_M = E_M \setminus E_C$.  

\textbf{Enumerating cyclic orientations} described in Section~\ref{sec:enum} exploits the property (which we will show later) that finding cyclic orientations of $G$ corresponds to finding particular orientations in $M'$, called \emph{extended cyclic orientations}, and of $C$, called \emph{legal orientations}. In the \textbf{for} loop, these orientations of $M'$ and $C$ are enumerated so as to find all the cyclic orientations of $G$. As we will see for the latter task, it is important to have $C$ of logarithmic size to guarantee our claimed delay.

\begin{algorithm}[t]
\KwIn{An undirected connected graph $G(V,E)$}
\KwOut{All the cyclic orientations $\ori{G}(V,\ori{E})$}
\smallskip
\textbf{Algorithm setup (Section~\ref{sec:setup}):}\\
Remove dead-ends (nodes of degree $1$) recursively from $G$\\
$M(V_M,E_M)\gets $ replace $G$'s maximal paths of degree-2 nodes with chain edges\\
$C(V_C,E_C) \gets $ a \loghole of $M$\\
$M'(V_M,E'_M)\gets $ delete the edges of $C$ from $M$, i.e.  $E'_M = E_M \setminus E_C$\\[2mm]
\textbf{Enumerate cyclic orientations (Section~\ref{sec:enum}):}\\
\For{\emph{each extended orientation $\ori{M}'$ of multigraph $M'$}}{
    \For{\emph{each legal orientation $\ori{C}$ of \loghole $C$ (see Algorithm~\ref{alg:ternpart})}}{
        $\ori{M}''(V_M,\ori{E}'') \gets $ combine $\ori{M}'$ and $\ori{C}$, where $\ori{E}'' = \ori{E}'_M\cup \ori{E}_C$ \\
        Output each of the cyclic orientations $\ori{G}$ of $G$ corresponding to  $\ori{M}''$\\
    }
}
\caption{Returning all the cyclic orientations of $G$}
\label{alg:structure}
\end{algorithm}

\section{Algorithm Setup}
\label{sec:setup}
\hide{%
In this section, given the undirected graph $G$, we show to compute an undirected labeled multigraph $M$ such that finding cyclic orientations of $G$ is reduced to finding particular orientations in $M$, called extended cyclic orientations. Our method for finding these orientations is explained in Section~\ref{sec:enum}, and makes use of the \loghole of $M$. hence, in the second phase of the preprocessing, the \loghole of $M$ is computed. The first preprocessing phase runs in $O(|E|)$; the second preprocessing phase has a cost that depends on the number of solutions will be produced by the algorithm: it runs in $O(s\cdot |E_M|)=O(s\cdot |E|)$, with $s\leq |V_M|\leq |V|$, for some $s$ such that the number of solutions is at least $2^s$. This means that our setup procedure is more time-consuming when we will produce a sufficiently large number of solutions, while it is almost linear otherwise, meaning that this cost is completely negligible in the overall cost of the algorithm. The space usage of the whole setup process is linear.
In Section~\ref{sec:absorb} we will see a way of computing the setup whose time cost is equal to the delay of the enumeration algorithm presented in Section~\ref{sec:enum} and whose space usage is $\Theta( |V|\cdot |E|)$.
}

\subsection{Reducing the problem to extended cyclic orientations}
\label{sec:rec}
In the following we show how to reduce Problem~\ref{problem1} to an extended version that allows us to neglect dead-ends and chains.


\paragraph{\bf Dead-end removal.}
Given an undirected graph $G(V,E)$, a dead-end is a node of degree 1.
We  recursively remove all \textit{dead-ends}, so that all the surviving nodes have degree 2 or greater. By removing these nodes and computing the cyclic orientations in the cleaned graph, we can still generate solutions for the original graph by using both orientations of the unique incident edge to each dead-end, as emphasized by the following lemma, whose proof is straightforward.
\begin{lemma}
Let $G$ be a graph, $u$ a dead-end and $\{u,v\}$ its unique incident edge. Let $G'$ be the graph $G$ without $u$ and the edge $\{u,v\}$. The set of all the cyclic orientations of $G$ is composed by the orientations $\ori{G}'(V'\cup\{u\},\ori{E}'\cup \{(u,v)\})$ and $\ori{G}'(V'\cup\{u\},\ori{E}'\cup \{(v,u)\})$, for all the cyclic orientations $\ori{G}'(V',\ori{E}')$ of $G'$.
\label{lem:deadends}
\end{lemma}
The dead-end removal can be done by performing a DFS recursive traversal of $G$, starting from an arbitrary node $x$. Every time a node of degree 1 is visited, it is removed from the graph. When the recursion ends in a node, the latter is removed if all of its neighbors have been removed except one (which is its parent in the DFS tree). Finally, when the traversal ends, it might be that the node from which we started has degree 1. To complete the process, if $x$ has now degree 1, we remove it from the graph. 
%
The DFS of $G$ and the removal of its dead-ends can be done in $O(m)$.

The rationale for removing dead-ends is to have shorter cycles: for example, consider a ``necklace'' graph with $n$ nodes, for $n$ even, such that $n/2$ nodes form a cycle, and the remaining $n/2$ nodes have degree~1 and are attached to one of the nodes in the cycle, such that each node in the cycle has degree~3 and is connected to one node of degree 1. With the removal of dead-ends, the cycle has only nodes of degree~2 and can be compressed as discussed in the next paragraph.

\paragraph{\bf Chain compression.}
It consists in finding  all the maximal paths $v_1,\ldots , v_k$ where $v_i$ has degree 2 (with $2\leq i\leq k-1$), and replacing each of them, called \emph{chain}, with just one edge, called \emph{chain edge}. It is easy to see that this task can be accomplished in $O(m)$ time by traversing the graph $G$ in a DFS fashion from a node of degree $\geq 3$.
The output is an undirected connected multigraph $M(V_M,E_M)$, where $V_M \subseteq V$ are the nodes of $V$ whose degree is $\geq 3$, and $E_M$ are the chain edges plus all the edges in $E$ which are not part of a chain.\footnote{The degenerate case of $M$ with $\leq 3$ nodes can be handled separately.} The latter ones are called simple edges to distinguish them from the chain edges. In the rest of the paper, $M$ will be seen as a multigraph where $|V_M| \geq 4$ and each of the edges has a label in \{simple, chain\}, since it might contain parallel edges or loops. For this, we define the concept of \emph{extended orientation} as follows.
\begin{definition}[Extended Orientation]
\emph{%
For a multigraph $M(V_M,E_M)$ having self loops and  edges labeled as simple and chain, an \emph{extended orientation} $\ori{M}(V_M,\ori{E_M})$ is an orientation $\ori{E_M}$
of its edge set $E_M$: for any simple edge $\{u,v\}$, exactly one direction between $(u,v)$ and
$(v,u)$ holds; for any chain edge $\{u,v\}$, either is \emph{broken}, or exactly one direction between $(u,v)$ and $(v,u)$ holds. A directed cycle in $\ori{M}$ cannot contain a broken edge.
}
\end{definition}
Broken edges correspond to chain edges that, when expanded as edges of $G$, do not have an orientation as a directed path. This means that they cannot be traversed in either direction. Note that this situation cannot happen for simple edges. The following lemma holds.
\begin{lemma}
\label{lem:ext}
If we have an algorithm that list all the extended cyclic orientations of $M(V_M,E_M)$ with delay $f(|E_M|)$, for some $f:\mathds{R}\to\mathds{R}$, then we have an algorithm that lists all the acyclic orientations of the graph $G(V,E)$ with delay $O(f(|E_M|) + |E|)$. 
\end{lemma}
\begin{proof}
For each extended cyclic orientations $\ori{M}$ we return a set $S$ of cyclic orientations of $G$: any edge $e$ of $\ori{M}$ maintains the same direction specified by $\ori{M}$ in all the solutions in $S$; for each chain $c$ of $\ori{M}$, we consider the edges corresponding to $c$ in $G$, say $e_1,e_2,\ldots, e_h$: if $c$ has a direction in $\ori{M}$, the same direction of $c$  is assigned  to all the edges $e_j$ in all the solutions in $S$; if $c$ has no direction assigned, i.e. \emph{broken}, we have to consider all the possible $2^{h}-2$ ways of making the path $e_1,e_2,\ldots, e_h$ broken (these are all the possible ways of directing the edges except the only two directing a path). All the solutions in $S$ differ for the way they replace the chain edges.

Getting extended cyclic orientations in $f(|E_M|)$ delay, iterating over all the chain edges $c$, and iterating over all the corresponding edges of $c$ assigning the specified directions as explained above, we return acyclic orientations of the graph $G(V,E)$ with delay $O(f(|E_M|) + |E|)$.\qed
\end{proof}

Lemma \ref{lem:ext} allows us to concentrate on extended cyclic orientations of the labeled multigraph $M$ rather than on cyclic orientations of $G$. Conceptually, we have to assign binary values (the orientation) to simple edges and ternary values (the orientation or broken) to chain edges. If we complicate the problem on one side by introducing these multigraphs with chain edges, we have a relevant benefit on the other side, as shown next.

\subsection{Logarithmically bounded hole}
\label{sec:girth}

Given the labeled multigraph obtained in Section~\ref{sec:rec}, namely $M(V_M,E_M)$, we perform the following two steps.
\begin{enumerate}
\item \textbf{Finding a \loghole.} Find a \emph{logarithmically bounded hole} (hereafter, \loghole) $C(V_C,E_C)$ in $M(V_M,E_M)$: it is a chordless cycle whose length is either the \emph{girth} of the graph (i.e. the length of its shortest cycle) or this length plus one.\footnote{Minimum cycle means the cycle having minimum number of edges (i.e. a  self loop), where chain edges count just one like normal edges.}
\item \textbf{Removing the \loghole.} Remove the edges in $E_C$ from $M$, obtaining $M'(V_M,E'_M)$, where $E'_M=E_M-E_C$.
\end{enumerate}

\subsubsection{Properties of the \loghole.} Since $M$ is a multigraph with self-loops, a \loghole $C(V_C,E_C)$ of $M$ can potentially be a self-loop. In any case,  the following well-known result holds.
\begin{lemma}[Logarithmic girth~\cite{B78,EP62}]
\label{prop:girth}
Let $G(V,E)$ be a graph in which every node has degree at least $3$. The girth of $G$ is $\leq 2\lceil\log |V|\rceil$.
\end{lemma}

Lemma~\ref{prop:girth} means that the \loghole $C$ of $M$ has length at most $2\lceil\log |V_M|\rceil+1$, thus motivating our terminology. 

The \loghole $C$ can be found by finding the girth, that is performing a BFS on each node of the multigraph $M$ to identify the shortest cycle that contains that node, in time $O(|V_M|\cdot |E_M|)$. By applying the algorithm in~\cite{itai1978finding}, which easily extends to multigraphs, we compute $C$ in time $O(|V_M|^2)$: in this case, if chords are present in the found $C$, in time $O(\log |V_M|)$ we can check whether $C$ includes a smaller cycle and redefine $C$ accordingly.

\section{Enumerating cyclic orientations}
\label{sec:enum}
We now want to list all the cyclic orientations of the input graph $G$. By Lemma~\ref{lem:ext} this is equivalent to listing the extended cyclic orientations of the corresponding labeled multigraph $M(V_M,E_M)$ obtained from $G$ by dead-end removal and chain compression. We now show that the latter task can be done by suitably combining some orientations from the labeled multigraph
 $M'(V_M,E'_M)$ and the \loghole $C(V_C,E_C)$ using an algorithm that is organized as follows.
\begin{enumerate}
\item \textbf{Finding extended orientations.} Enumerate all extended orientations (not necessarily cyclic) $\ori{M'}$ of the multigraph $M'$.
\item \textbf{Putting back the \loghole.} For each listed $\ori{M'}(V_M,\ori{E'}_M)$, consider all the extended orientations $\ori{C}(V_C,\ori{E}_C)$ of the \loghole $C$ such that $\ori{E'}_M\cup \ori{E}_C$ contains a directed cycle, and obtain the extended cyclic orientations for the multigraph $M$.
\end{enumerate}

\subsubsection{Finding extended orientations.}
This is now an easy  task. For each edge $\{u,v\}$ in $E'_M$ that is labeled as simple, both the directions $(u,v)$ and $(v,u)$ can be assigned; if $\{u,v\}$ is labeled as chain, the directions $(u,v)$ and $(v,u)$, and broken can be assigned. Each combination of these decisions produces an extended orientation of $M'(V_M,E_M)$. If there are $s$ simple edges and $b$ broken edges in $M'$, where $s+m = |E'_M|$, this generates all possible $2^s 3^b$ extended orientations.   Each of them can be easily listed in $O(|E'_M|)$ delay (actually less, but this is not the dominant cost).

\subsubsection{Putting back the \loghole.}
For each listed $\ori{M'}$ we have to decide how to put back the edges of the cycle $C$, namely, how to find the orientations of $C$ that create directed cycles. 


\begin{definition}
Given the cycle $C(V_C,E_C)$ and $\ori{M}'(V_M,\ori{E}'_M)$, we call \emph{legal orientation} $\ori{C}(V_C,\ori{E}_C)$ any extended orientation of $C$ such that the resulting multigraph $\ori{M}''(V_M, \ori{E}'')$ is cyclic, where $\ori{E}''=\ori{E}'_M\cup \ori{E}_C$. 
\end{definition}
%
The two following cases are possible.
\begin{enumerate}
\item $\ori{M}'$ is cyclic. In this case each edge in $E_C$ can receive any direction, including broken if the edge is a chain edge: each combination of these assignments will produce a legal orientation that will be output.
\item $\ori{M}'$ is acyclic. Since $C$ is a cycle, there are at least two legal orientations obtained by orienting $C$ as a directed cycle clockwise and counterclockwise. Moreover, adding just an oriented subset of edges $D \subseteq C$ to $\ori{M}'$ can create a cycle in $\ori{M}'$: in this case, any orientation of the remaining edges of $C \setminus D$ (including broken for chain edges) will clearly produce a legal orientation.
\end{enumerate}

While the first case is immediate, the second case has to efficiently deal with the following problem.
\begin{problem}
\label{prob:reintegrate}
Given $\ori{M'}$ acyclic and cycle $C$, enumerate all the legal orientations $\ori{C}(V_C,\ori{E}_C)$ of $C$.
\end{problem}

In order to solve Problem~\ref{prob:reintegrate}, we exploit the properties of $C$. In particular, we compute the reachability matrix $R$ among all the nodes in $V_C$, that is, for each pair $u,v$ of nodes in $V_C$, $R(u,v)$ is $1$ if $u$ can reach $v$ in the starting $\ori{M}$, $0$ otherwise. We say that $R$ is \emph{cyclic} whether there exists a pair $i,j$ such that $R(i,j)=R(j,i)=1$. This step can be done by performing a BFS in $\ori{M'}$ from each node in $V_C$: by \Cref{prop:girth} we have $|V_C|<2\lceil\log |V_M|\rceil+1$, and so the cost is $O(|E'_M| \cdot \log |V_M|)$ time. Deciding the orientation of the edges and the chain edges in $E_C$ is done with a ternary partition of the search space. Namely, for each edge $\{u,v\}$ in $E_C$, if $\{u,v\}$ is simple we try the two possible direction assignments, while if it is a chain we also try the broken assignment. In order to avoid dead-end recursive calls, after each assignment we update the reachability matrix $R$ and perform the recursive call only if this partial direction assignment will produce at least a solution: both the update of $R$ and the dead-end check can be done in $O(\log^{2}|V_M|)$ (that is, the size of $R$).

\paragraph*{\bf Scheme for legal orientations.}
The steps are shown by Algorithm~\ref{alg:ternpart}. At the beginning the reachability matrix $R$ is computed and is passed to the recursive routine \texttt{LegalOrientations}. At each step,  $\ori{C}'$ is the partial legal orientation to be completed and $I$ is the set of broken edges declared so far. Also, $j$ is the index of the next edge $\{c_{j},c_{j+1}\}$ of the cycle $C$,  with $1\leq j\leq h$ (we assume $c_{h+1}=c_1$ to close the cycle): if $j=h+1$ then all the edges of $C$ have been considered and we output the solution $\ori{C}'$ together with the list $I$ of broken edges in $\ori{C}'$.
Each time the procedure is called we guarantee that the reachability matrix $R$  is updated.

Let $\{u,v\}$ be the next edge of $C$ to be considered: for each possible direction assignment $(u,v)$ or $(v,u)$ of this edge, we have to decide whether we will be able to complete the solution considering this assignment. This is done by trying to add the arc to the current solution. If there is already a cycle, clearly we can complete the solution. Otherwise, we perform a \textit{reachability check} on $\{c_{j+1},\ldots c_{h+1}\}$: it is still possible to create a directed cycle if and only if any two of the nodes in $\{c_{j+1},\ldots c_{h+1}\}$, say $c_f$ and $c_g$ satisfy $R(c_f,c_g) = 1$ or $R(c_g,c_f) =1$. This condition guarantees that a cycle will be created in the next calls, since we know there are edges in $C$ between $c_f$ and $c_g$ that can be oriented suitably. Finally, when $\{u,v\}$ is a chain, the broken assignment is also considered: $R$ does not need to be updated as the broken edge does not change the reachability of $\ori{M'}$.

The reachability and cyclicity checks are done by updating and checking the reachability matrix $R$ (and restoring $R$ when needed). Updating $R$ when adding an arc $(u,v)$ corresponds to making $v$, and all nodes reachable from $v$, reachable from $u$ and nodes that can reach $u$. This can be done by simply performing an \texttt{or} between the corresponding rows in time $O(\log^2|V_M|)$, since $R$ is $|C|\times|C|$. The reachability check can be done in  $O(\log^{2}|V_M|)$ time. The cyclicity (checking whether a cycle has been already created) takes the same amount of time by  looking for a pair of nodes in $\{c_{1},\ldots c_{j}\}$ $x'$,$y'$  such $R(x',y')=R(y',y')=1$.

\begin{lemma}
Algorithm~\ref{alg:ternpart} outputs in $O(|E'_M|\log |V_M|)$ time the first legal orientation of $C$, and each of the remaining ones with $O(\log^3 |V_M|)$ delay.
\label{lem:delaytern}
\end{lemma}
\begin{proof}
Before calling the \texttt{LegalOrientations} procedure we have to compute the reachability matrix from scratch and this costs $O(|E'_M|\log |V_M|)$ time. In the following we will bound the delay between two outputs returned by the  \texttt{LegalOrientations} procedure.
Firstly, note that each call produces at least one solution. This is true when $j=1$ since we have two possible legal orientations of $C$. Before performing any call at depth $j$, the caller function checks whether this will produce at least one solution. Only calls that will produce at least one solution are then performed. This means that in the recursion tree, every internal node has at least a child and each leaf corresponds to a solution. Hence the delay between any two consecutive solutions is bounded by the cost of a leaf-to-root path and the cost of a root-to-the-next-leaf path in the recursion tree induced by \texttt{LegalOrientations}. Since the height of the recursion tree is $O(\log |V_M|)$, i.e. the edges of $C$, and the cost of each recursion node is $O(\log^2 |V_M|)$, we can conclude that the delay between any two consecutive solutions is bounded by $O(\log^3 |V_M|)$. As it can be seen, it is crucial that the size of $C$ is (poly)logarithmic.
\qed
\end{proof}

\begin{algorithm}[t]
\SetKwFunction{proc}{LegalOrientations}
\SetKwFunction{ccheck}{AllowsCycle}
\KwIn{$\ori{M'}(V_M,\ori{E'}_M)$ acyclic, a cycle $C(V_C,E_C)$ with $V_C\subseteq V_M$}
\KwOut{All the legal orientations $\ori{C}(V_C,\ori{E}_C)$ }
\smallskip
Build the reachability matrix $R$ for the nodes of $V_C$ in $\ori{M'}$\\
Let $V_C=\{c_1,\ldots c_h\}$, where $c_{h+1}=c_1$ by definition\\
Execute \proc $(\ori{C}'(\emptyset,\emptyset), 1, R, \emptyset)$\\[3mm]
\setcounter{AlgoLine}{0}
  \SetKwProg{myproc}{Procedure}{}{}
  \myproc{\proc{$\ori{C}'(V'_C,\ori{E}'_C), j, R, I$}}{
    
    
    \eIf{$j=h+1$ }
        {output $\ori{C}'$ and its set $I$ of broken edges}
        {
            $u \gets c_j$, \: $v \gets c_{j+1}$ 
            
            $R_1 \leftarrow R$ updated by adding the arc $(u,v)$\;
            \If{$R_1$ is cyclic or has positive reachability test on $\{c_{j+1},\ldots, c_{h+1}\}$}{
                \proc ($\ori{C}'(V'_C,\ori{E}'_C\cup \{(u,v)\}), j+1, R_1, I)$\\
            }
            
            $R_2 \leftarrow R$ updated adding the arc $(v,u)$\;
            \If{$R_2$ is cyclic or has positive reachability test on $\{c_{j+1},\ldots, c_{h+1}\}$}{
                \proc ($\ori{C}'(V'_C,\ori{E}'_C\cup \{(v,u)\}), j+1, R_2, I)$\\
            }
            
            \If{$\{u,v\}$ is a chain edge}
            {
            \If{$R$ is cyclic or has positive reachability test on $\{c_{j+1},\ldots, c_{h+1}\}$}{
                \proc ($\ori{C}', j+1, R, I \cup\{\{u,v\}\})$\\ 
                }
            }
        }
}
\caption{Returning all legal orientations of $C$}
\label{alg:ternpart}
\end{algorithm}





\begin{lemma}[Correctness]
\begin{enumerate}
\item All the extended cyclic orientations of $M$ are output. 
\item Only extended cyclic orientations of $M$ are output.
\item There are no duplicates.
\end{enumerate}
\end{lemma}
\begin{proof}
\begin{enumerate}
\item Any extended cyclic orientation $\ori{M}$ can be seen as the union $\ori{M}''$ of $\ori{M}'$ and $\ori{C}$, which are two edge disjoint directed subgraphs. Our algorithm enumerates all the extended orientations of $M'$, and, for each of them, all legal extended orientations $\ori{C}$: if there is a cycle in $\ori{M}$ involving only edges in $E'_M$ all the extended orientations of $C$ are legal; otherwise just the extended orientations $\ori{C}$ of $C$ whose arcs create a cycle in $\ori{M}''$ are legal. Hence any extended cyclic orientation $\ori{M}$ is output.
\item Any output solution is an extended orientation: each edge in $M'$ and in $C$ has exactly one direction or is broken. Moreover, any output solution contains at least a cycle: if there is not a cycle in $M'$, a cycle is created involving the edges in $C$.
\item All the extended orientations $\ori{M}=\ori{M}''$ in output differ for at least an arc in $\ori{E}'_M$ or an arc in $\ori{E_C}$. Hence there are no duplicate solutions.
\end{enumerate}
\qed
\end{proof}

As a result, we obtain delay $\tilde{O}(|E_M|)$.
\begin{lemma}
The extended cyclic orientations of $M(V_M,E_M)$ can be enumerated with delay $\tilde{O}(|E_M|)$ and space $O(|E_M|)$.
\label{lem:delay2}
\end{lemma}
\begin{proof}
Finding extended orientations $\ori{M'}$ of $M'$ can be done with $O(|E'_M|)$ delay. Every time a new $\ori{M'}$ has been generated, we apply Algorithm~\ref{alg:ternpart}. By Lemma~\ref{lem:delaytern} we output the first cyclic orientation $\ori{M}$ of $M$ with delay $O(|E_M|\log |V_M|)$ and the remaining ones with delay $O(\log^3 |V_M|)$. Hence the maximum delay between any two consecutive solutions is $O(|E'_M|+|E_M|\log |V_M|)=O(|E_M|\log |V_M|)=\tilde{O}(|E_M|)$. The space usage is linear in all the phases: in particular in Algorithm~\ref{alg:ternpart} the space is $O(\log^2 |V_M|)$, because of the reachability matrix $R$, which is smaller than $O(|E_M|)$.
\qed
\end{proof}

Applying \Cref{lem:delay2} and \Cref{lem:ext}, and considering the setup cost in Section~\ref{sec:setup} ($|V_M|\leq |V|$ and $|E_M|\leq |E|$), we can conclude as follows.

\begin{theorem}
\label{the:delay}
Algorithm~\ref{alg:structure} lists all cyclic orientations of $G(V,E)$ with setup cost $O(|V|^2)$,
and delay $\tilde{O}(|E|)$. The space usage is $O(|E|)$ memory cells.
\end{theorem}

\section{Absorbing the setup cost}
\label{sec:absorb}

In this section, we show how to modify our approach to get a setup time equal to the delay, requiring space $\Theta(|V|\cdot |E|)$.

\begin{theorem}
\label{the:absorb}
All cyclic orientations of $G(V,E)$ can be listed with setup cost $\tilde{O}(|E|)$, delay $\tilde{O}(|E|)$, and space usage of $\Theta(|V|\cdot |E|)$ memory cells.
\end{theorem}

We use $n = |V|$ and $m = |E|$ for brevity.  Let $A_1$ be the
following algorithm that takes $T_1 = O(mn)$ time to generate $n$
solutions, each with $\tilde{O}(m)$ delay, starting from any given
cycle of size $\geq \log n$. This cycle is found by performing a BFS
on an arbitrary node $u$, and identifying the shortest cycle $C_u$
containing $u$. Note that $C_u$ is a \loghole as required. Now, if
$|C_u| < \log n$, we stop the setup and run the algorithms in the
previous sections setting $C=C_u$. The case of interest in this
section is when $|C_u| \geq \log n$. We take a cyclic orientation
$\ori{C_u}$ of $C_u$, and then $n$ arbitrary orientations of the edges
in $G \setminus C_u$.  The setup cost is $O(m)$ time and we can easily
output each solution in $\tilde{O}(m)$ delay. We denote this set of
$n$ solutions by $Z_1$.

Also, let $A_2$ be the algorithm behind Theorem~\ref{the:delay}, with a setup cost of $O(mn)$ and  $\tilde{O}(m)$ delay (i.e. Algorithm~\ref{alg:structure}). We denote the time taken by $A_2$ to list the first $n$ solutions, including the $O(mn)$ setup cost, by $T_2 = \tilde{O}(mn)$, and this set of $n$ solutions by $Z_2$. Since $Z_1$ and $Z_2$ can have nonempty intersection, we want to avoid duplicates. 

We show how to obtain an algorithm $A$ that lists all the cyclic orientations without duplicates with $\tilde{O}(m)$ setup cost and delay, using $O(mn)$ space. Even though the delay cost of $A$ is larger than that of $A_1$ and $A_2$ by a constant factor, the asymptotic complexity is not affected by this constant, and remains $\tilde{O}(m)$. 

Algorithm $A$ executes simultaneously and independently the two algorithms $A_1$ and $A_2$. Recall that these two algorithms take $T_1+T_2$ time in total to generate $Z_1$ and $Z_2$ with $\tilde{O}(m)$  delay. However those in $Z_2$ are produced after a setup cost of $O(mn)$.  Hence $A$ slows down on purpose by a constant factor $c$, thus requiring $c (T_1+T_2)$ time: it has time to find the distinct solutions in $Z_1 \cup Z_2$ and build a dictionary $D_1$ on the solutions in $Z_1$. (Since an orientation can be represented as a binary string of length $m$, a binary trie can be employed as dictionary $D_1$, supporting each dictionary operation in $O(m)$ time.) During this time, $A$ outputs the $n$ solutions from $Z_1$ with a delay of $c (T_1+T_2)/n = \tilde{O}(m)$ time each, while storing the rest of solutions of $Z_2\setminus Z_1$ in a buffer $Q$. 

After $c (T_1+T_2)$ time, the situation is the following: Algorithm $A$ has output the $n$ solutions in $Z_1$ with $\tilde{O}(m)$ setup cost and delay. These solutions are stored in $D_1$, so we can check for duplicates. We have buffered at most $n$ solutions of $Z_2\setminus Z_1$ in $Q$. Now the purpose of $A$ is to continue with algorithm $A_2$ alone, with $\tilde{O}(m)$ delay per solution, avoiding duplicates. Thus for each solution given by $A_2$, algorithm $A$ suspends $A_2$ and waits so that each solution is output in $c (T_1+T_2)/n$ time: if the solution is not in $D_1$, $A$ outputs it; otherwise $A$ extracts one solution from the buffer $Q$ and outputs the latter instead. Note that if there are still $d$ duplicates to handle in the future, then $Q$ contains exactly $d$ solutions from  $Z_2\setminus Z_1$ (and $Q$ is empty when $A-2$ completes its execution). Thus, $A$ never has to wait for a non-duplicated solution. The delay is the maximum between $c (T_1+T_2)/n$ and the delay of $A_2$, hence $\tilde{O}(m)$. The additional space is dominated by that of $Q$, namely, $O(mn)$ memory cells to store up to $n$ solutions.

We also have an amortized cost using the lemma below, where $f(x) = \tilde{O}(x)$ and $s = |V|$.

\begin{lemma}
Listing all the extended cyclic orientations of $M(V_M,E_M)$ with delay $O(f(|E_M|))$ and setup cost $O(s \cdot |V_M|)$ implies that the average cost per solution is $O(f(|E_M|)+ |E_M|)$.
\label{lem:preproc}
\end{lemma}
\begin{proof}
We perform a BFS on an arbitrary node $u$, and identify the shortest cycle $C_u(V_u,E_u)$ that contains $u$. This costs $O(m)$ time. 
Note that $C_u(V_u,E_u)$ is a hole (i.e. it has no chords). Note that a minimum cycle in $M$ either is $C_u$ or contains a node in $V_M-V_u$: hence we perform all the BFSs from each node in $V_M-V_u$, as explained in~\cite{itai1978finding} with an overall cost of $O(|V_M|\cdot |V_M-V_u|)$. The number of extended orientations of $M$ is at least $2^{|E_M-E_u|}\geq 2^{|V_M-V_u|}$.
Our setup cost is $O(s \cdot |V_M|)$, with $s\leq |V_M|$, and the number of solutions is at least $2^s$. The overall average cost per solution is at most
\[\frac{O(2^s\cdot f(|E_M|) + s \cdot |V_M| )}{2^s}=O\left(f(|E_M|) + |E_M|\cdot \frac{s}{2^s}\right)\]
\qed
\end{proof}

\section{Conclusions}
\label{sec:conclusions}

In this paper we considered the problem of efficiently enumerating cyclic orientations of graphs. The problem is interesting from a combinatorial and algorithmic point of view, as the fraction of cyclic orientations over all the possible orientations can be as small as 0 or very close to 1. We provided an efficient algorithm to enumerate the solutions with delay $\tilde{O}(m)$ and overall complexity $\tilde{O}(\alpha\cdot m)$, with $\alpha$ being the number of solutions.

\bibliographystyle{abbrv}

\begin{thebibliography}{10}

\bibitem{alon1995acyclic}
N.~Alon and Z.~Tuza.
\newblock The acyclic orientation game on random graphs.
\newblock {\em Random Structures \& Algorithms}, 6(2-3):261--268, 1995.

\bibitem{Barbosa199971}
V.~C. Barbosa and J.~L. Szwarcfiter.
\newblock Generating all the acyclic orientations of an undirected graph.
\newblock {\em Information Processing Letters}, 72(1):71 -- 74, 1999.

\bibitem{B78}
B.~Bollobas.
\newblock {\em Extremal Graph Theory}.
\newblock Dover Publications, Incorporated, 2004.

\bibitem{EP62}
P.~Erd{\H{o}}s and L.~P{\'o}sa.
\newblock On the maximal number of disjoint circuits of a graph.
\newblock {\em Publ. Math. Debrecen}, 9:3--12, 1962.

\bibitem{Fisher199773}
D.~C. Fisher, K.~Fraughnaugh, L.~Langley, and D.~B. West.
\newblock The number of dependent arcs in an acyclic orientation.
\newblock {\em Journal of Combinatorial Theory, Series B}, 71(1):73 -- 78,
  1997.

\bibitem{itai1978finding}
A.~Itai and M.~Rodeh.
\newblock Finding a minimum circuit in a graph.
\newblock {\em SIAM Journal on Computing}, 7(4):413--423, 1978.

\bibitem{JohnsonP88}
D.~S. Johnson, C.~H. Papadimitriou, and M.~Yannakakis.
\newblock On generating all maximal independent sets.
\newblock {\em Inf. Process. Lett.}, 27(3):119--123, 1988.

\bibitem{linial1986hard}
N.~Linial.
\newblock Hard enumeration problems in geometry and combinatorics.
\newblock {\em SIAM Journal on Algebraic Discrete Methods}, 7(2):331--335,
  1986.

\bibitem{moon1968topics}
J.~Moon.
\newblock {\em Topics on tournaments}.
\newblock Athena series: Selected topics in mathematics. Holt, Rinehart and
  Winston, 1968.

\bibitem{CPC:6776456}
O.~Pikhurko.
\newblock Finding an unknown acyclic orientation of a given graph.
\newblock {\em Combinatorics, Probability and Computing}, 19:121--131, 1 2010.

\bibitem{Richardson:1996:DAD:2074284.2074338}
T.~Richardson.
\newblock A discovery algorithm for directed cyclic graphs.
\newblock In {\em Proceedings of the Twelfth International Conference on
  Uncertainty in Artificial Intelligence}, UAI'96, pages 454--461, San
  Francisco, CA, USA, 1996. Morgan Kaufmann Publishers Inc.

\bibitem{Spirtes:1995:DCG:2074158.2074214}
P.~Spirtes.
\newblock Directed cyclic graphical representations of feedback models.
\newblock In {\em Proceedings of the Eleventh Conference on Uncertainty in
  Artificial Intelligence}, UAI'95, pages 491--498, San Francisco, CA, USA,
  1995. Morgan Kaufmann Publishers Inc.

\bibitem{Squire1998275}
M.~B. Squire.
\newblock Generating the acyclic orientations of a graph.
\newblock {\em Journal of Algorithms}, 26(2):275 -- 290, 1998.

\bibitem{stanley87}
R.~Stanley.
\newblock Acyclic orientations of graphs.
\newblock In I.~Gessel and G.-C. Rota, editors, {\em Classic Papers in
  Combinatorics}, Modern Birkh\"auser Classics, pages 453--460. Birkh\"auser
  Boston, 1987.

\end{thebibliography}

\end{document}